\documentclass[conference]{IEEEtran}
\usepackage{amsmath, amsthm, enumerate, tikz}
\usepackage[hidelinks]{hyperref}

\newtheorem{proposition}{Proposition}

\theoremstyle{definition}
\newtheorem{definition}{Definition}

\theoremstyle{remark}

\newtheorem{example}{Example}

\newcommand{\defeq}{\mathrel{\mathop :} =}
\newcommand{\E}{\,1\,}
\newcommand{\R}{\,\rm R\,}\newcommand{\B}{\,\rm B\,}
\newcommand{\C}{\,\rm C\,}\newcommand{\M}{\,\rm M\,}\newcommand{\Y}{\,\rm Y\,}
\renewcommand{\O}{\,0\,}

\title{A Monoid Ring Approach\\ to Color Visual Cryptography}

\author{\IEEEauthorblockN{Maximilian Reif and Jens Zumbrägel}
  \IEEEauthorblockA{%
    \textit{Faculty of Computer Science and Mathematics}\\
    \textit{University of Passau}\\
    Innstraße 33, 94032 Passau, Germany}\\
    \{maximilian.reif, jens.zumbraegel\}@uni-passau.de}

\begin{document}

\maketitle

\begin{abstract}
  A visual cryptography scheme is a secret sharing scheme in which
  the secret information is an image and the shares are printed on
  transparencies, so that the secret image can be recovered by simply
  stacking the shares on top of each other.  Such schemes do therefore
  not require any knowledge of cryptography tools to recover the
  secret, and they have widespread applications, for example, when
  sharing QR codes or medical images.

  In this work we deal with visual cryptography threshold schemes for
  color images.  Our color model differs from most previous work by
  allowing arbitrary colors to be stacked, resulting in a possibly
  different color.  This more general color monoid model enables us to
  achieve shorter pixel expansion and higher contrast than comparable
  schemes.  We revisit the polynomial framework of Koga and Ishihara
  for constructing visual cryptography schemes and apply the monoid
  ring to obtain new schemes for color visual cryptography.
\end{abstract}

\section{Introduction}

Visual cryptography introduced by Naor and Shamir~\cite{NS94} allows
to share a secret image among~$n$ participants, so that at
most~$t \!-\! 1$ participants gain no information about the image
while any~$t$ participants recover the image visually by stacking
their~$t$ shares one upon the other.  Originally devised for black and
white pictures, there has also been interest in color versions of this
problem from early on.

In color visual cryptography there are various models of how different
colors are combined when stacked together.  In the model by Verheul
and Tilborg~\cite{VT97} different colors other than black are not
allowed to be stacked, while the composition of a color with black
results in black, and with the same color returns this color.  A more
general model by Koga and Yamamoto~\cite{KY98} describes colors as
elements in a lattice with the combination of colors is given by the
join in this lattice.  Some recent works on constructing visual
cryptography schemes with algebraic methods include papers by
Adhikari~\cite{A14}, Dutta, Adhikari and Ruj~\cite{DAR19}, Dutta et
al.~\cite{D+20}, and Koga~\cite{K22}.

Visual cryptography schemes are usually described by certain
$n \times q$-matrices referred to as basis matrices, one for each
possible color of a pixel, where~$q$ denotes the pixel expansion.
Every pixel is expanded in the $i$-th share as~$q$ subpixels
represented by the $i$-th row of the corresponding basis matrix,
applying a random column permutation.  Koga and Ishihara~\cite{KI11}
introduced an algebraic framework how to construct basis matrices.  It
is based on degree~$n$ homogeneous multivariate polynomials in the
colors as variables, which are expanded into basis matrices by
permutations of the corresponding columns.  The hiding property when
restricted to $t \!-\! 1$ rows is described elegantly by a
differential operator.

In this work we revisit the framework of~\cite{KI11} and apply it to a
color model based on a general monoid.  This way we can describe the
recovery process for the secret image by the operation in a monoid
ring, thereby enhancing the algebraic model.  We hope that our monoids
provide a more realistic model for color visual cryptography and
illustrate this by concrete examples.

\subsubsection*{Motivation for the extended color monoid model}

By making full use of the monoid operations we increase the
possibilities for the colors in the shares.  For example, if the
original image color is blue this may be stacked either by blue and
white, or by cyan and magenta.

These additional possibilities seem to have received little attention
in previous work.  However, they should enable better, more realistic
results for the image reconstruction by providing lower pixel
expansion and higher contrast.  The current work aims to investigate
these possibilities.

\subsubsection*{Outline of the paper}

In Section~\ref{sec:algebraic} we introduce the color monoid model and
recall the algebraic framework of Koga and Ishihara, which describes
visual cryptography schemes using multivariate polynomials.  In
Section~\ref{sec:construction} we present construction of basis
matrices for various $(t, n)$-threshold schemes that make use of the
full extended color monoid model.  Finally, in Section~%
\ref{sec:conclusion} we conclude the paper and discuss some
future research directions.

\section{The Algebraic Framework}%
\label{sec:algebraic}

In this section we introduce the algebraic framework that we use for
the color model, and we define the color visual cryptography schemes
we are considering.  We also provide a short review of the polynomial
approach for constructing basis matrices by Koga and Ishihara~%
\cite{KI11}.

\subsection{The Color Model}

We model the colors of the considered visual cryptography schemes as
elements of a finite \emph{commutative monoid}~$M$, which associates
to each pair of colors $(c_1, c_2) \in M \times M$ another color
$c_1 * c_2 \in M$ (we use multiplicative notation here in order to
facilitate the monoid ring later).  Specifically, we assume the
following properties for this operation:

\begin{enumerate}[\quad i)]
  \item The operation is commutative, i.e.
  \[ c_1 * c_2 \,=\, c_2 * c_1 \qquad \text{for all}~c_1, c_2 \in M . \]
  \item The operation is associative, i.e.
  \[ (c_1 * c_2) * c_3 \,=\, c_1 * (c_2 * c_3)
    \qquad \text{for all}~c_1, c_2, c_3 \in M . \]
\end{enumerate}
These properties i) and ii) ensure that the order in which the shared
colors $c_1, \dots, c_n$ are stacked is irrelevant, and we may write
$c_1 * \dots * c_n$ for the result.
\begin{enumerate}[\quad i)]\setcounter{enumi}2
  \item There is a neutral element $1 \in M$ such that
  \[ 1 * c \,=\, c \qquad \text{for all}~c \in M , \]
  which may be interpreted as the color “white”.
\end{enumerate}

Moreover, the commutative monoid~$M$ usually also fulfills the
following properties:
\begin{enumerate}[\quad i)]\setcounter{enumi}3
  \item There is an absorbing element $0 \in M$ such that
  \[ 0 * c \,=\, 0 \qquad \text{for all}~c \in M , \]
  which may be viewed as the color “black”.
  \item The operation is idempotent, i.e.
  \[ c * c \,=\, c \qquad \text{for all}~c \in M .  \]
\end{enumerate}
This last property corresponds to the non-darkening model, in which
stacking the same color multiple times does not darken its hue.  When
this condition is fulfilled also we speak of a \emph{semilattice}~%
\cite{KY98} and write $c_1 \wedge c_2$ for $c_1 * c_2$.  Here,
$c_1 \wedge c_2$ is called the infimum or meet of the semilattice.
(In fact, it is even a \emph{lattice}, i.e. for $c_1, c_2 \in M$ there
also exists the supremum or join $c_1 \vee c_2$.)

\begin{example}
  These are commutative monoids describing common color models.
  Each one fulfills the properties i) to v).
  \begin{enumerate}
  \item Let $M \defeq \{ 0, 1 \}$ with $c_1 \wedge c_2 \defeq c_1
    \text{ and } c_2$ (logical and).  This corresponds to the
    classical case of black and white visual cryptography, where~$0$
    is “black” and~$1$ is “white”.
  \item Let $M \defeq \{ 0, c_1, \dots, c_k, 1 \}$ where
    \[ c_p \wedge c_q \,\defeq\, \begin{cases} c_p &\text{if } p = q \\
        0 &\text{otherwise} \end{cases} \] as well as
    $c_j \wedge 1 = c_j$ and $c_j \wedge 0 = 0$ for all~$j$.  This
    corresponds to the model of Verheul and Tilborg~\cite{VT97}, with
    the slight variation that the result of different colors is black
    (rather than undefined).
  \item Let $M \defeq \{ 0, {\rm R}, {\rm G}, {\rm B}, {\rm Y}, {\rm M},
    {\rm C}, 1 \}$ be equipped with the operation as illustrated by
    the Hasse diagram in Fig.~\ref{fig:ymcrgb}, i.e. ${\rm Y} \wedge
    {\rm M} = {\rm R}$, ${\rm Y} \wedge {\rm C} = {\rm G}$ and
    ${\rm M} \wedge {\rm C} = {\rm B}$.  This models the colors
    Yellow, Magenta, Cyan which when combined produce the colors Red,
    Green, Blue.  In this work we will mainly consider this monoid or
    a substructure thereof.
    \begin{figure}
      \centering
      \begin{tikzpicture}[semithick, font=\footnotesize, shorten <=-1pt,
        shorten >=-1pt, x=1.25cm]
        \node (0) at (0, 0) { $0$ }; 
        \node (1) at (-1.5, 1.5) { $\rm R$ };
        \node (2) at (0, 1.5) { $\rm G$ };
        \node (3) at (1.5, 1.5) { $\rm B$ };
        \node (4) at (-1.5, 3) { $\rm Y$ };
        \node (5) at (0, 3) { $\rm M$ };
        \node (6)  at (1.5, 3) { $\rm C$ };
        \node (7) at (0, 4.5) { $1$ };
        \draw (0) -- (1);  \draw (0) -- (2);  \draw (0) -- (3);
        \draw (1) -- (4);  \draw (1) -- (5);  \draw (2) -- (4);
        \draw (2) -- (6);  \draw (3) -- (5);  \draw (3) -- (6);
        \draw (4) -- (7);  \draw (5) -- (7);  \draw (6) -- (7);
      \end{tikzpicture}
      \caption{A lattice for the colors Yellow, Magenta, Cyan, Red,
        Green, and Blue.}
      \label{fig:ymcrgb}
    \end{figure}
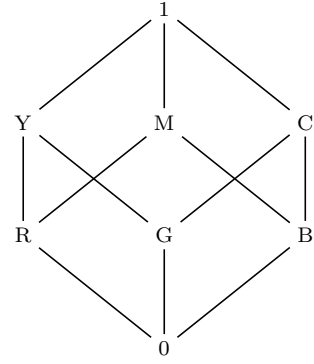
  \item Let $M \defeq \{ 0, c_1, \dots, c_k, 1 \}$ where the colors
    are linearly ordered so that $c_p \wedge c_q = c_p$ whenever
    $p \le q$.  This model describes a grayscale with shades of gray
    $c_1, \dots, c_k$ going from darker hue to brighter.
  \end{enumerate}
\end{example}

\subsection{Visual Cryptography Schemes}

Let us fix a finite monoid~$M = \{ c_0, c_1, \dots, c_k, c_{k+1} \}$
with commutative operation~$*$, neutral element $1 = c_{k+1}$ for
“white”, and assume an absorbing element $0 = c_0$ for “black” (we do
not require the monoid to be idempotent).  The elements of~$M$ model
the available colors, and $c_1 * c_2$ describes the result of stacking
the colors~$c_1$ and~$c_2$ together.  We now define the color visual
cryptography schemes by adapting the definition in~\cite[Def.~1]{KI11}.
As usual we use the notation $[n] \defeq \{ 1, \dots, n \}$.

\begin{definition}
  A \emph{visual cryptography scheme} for~$n$ participants and
  threshold~$t$ with pixel expansion~$q$ consists of $n \times q$-%
  matrices $S_1, \dots, S_{k+1}$ with entries in~$M$, such that the
  following properties are fulfilled.
  \begin{enumerate}[\quad i)]
  \item For every subset $A \subseteq [n]$ of cardinality $\vert A
    \vert \le t \!-\! 1$ the matrices $S_1[A], \dots, S_{k+1}[A]$ --
    consisting of the rows with index in~$A$ -- comprise the same
    columns up to permutation.
  \item For each $j = 1, \dots, k \!+\! 1$ and every $A \subseteq
    [n]$ of cardinality $\vert A \vert = t$ it holds that up to column
    permutation \[ S_j[A] \,=\, ( C_j \,\|\, Z_j \,\|\, R ) \] is the
    concatenation of matrices $C_j$, $Z_j$ and $R$ where
    \begin{itemize}
    \item for every column of~$C_j$ the product of entries (in the
      monoid~$M$) equals the color~$c_j$,
    \item each column of~$Z_j$ has a~$0$ (and thus product~$0$),
    \item the matrix~$R$ does not contain a~$0$ nor depend on~$j$.
    \end{itemize}
  \end{enumerate}    
\end{definition}

The matrices $S_1, \dots, S_{k+1}$ in the preceding definition are
referred to as basis matrices for the visual cryptography scheme.  To
distribute a pixel of color $c_j$ among the $n$ participants, the
$i$-th row of a random column permutation of $C_j$ is intended to be
used for the $i$-th share.  The first property i) is purely
combinatorial (it does not depend on the color model or any features
of the visual setup), and ensures that no fewer than~$t$ participants
are able to combine their shares to obtain any information about the
secret image.  On the other hand, property ii) enables every set
of~$t$ participants to recognize the color~$c_j$ by simply stacking
their shares together using the submatrix~$C_j$, while the column sums
of the submatrices~$Z_j$ and~$R$ produce black and random noise.

In a slight variation of this definition considered in~\cite{KI11},
one may ask to have the basis matrices for only a subset of the colors
$c_1, \dots, c_{k+1}$ (which may exclude $c_{k+1} = 1$).

\begin{example}
  Let us consider the submonoid of the lattice from Fig.~\ref{fig:ymcrgb}
  consisting of $M \defeq \{ 0, {\rm M}, {\rm C}, {\rm B}, 1 \}$ with
  ${\rm M} \wedge {\rm C} = {\rm B}$, meaning that Magenta and Cyan
  combine to Blue.

  The following are basis matrices for a simple visual cryptography
  scheme with $n = 2$ participants, threshold $t = 2$, and pixel
  expansion $q = 6$.
  \begin{gather*}
    S_1 \defeq \begin{bmatrix} \M \E \C \O \E \O \\
      \E \M \O \C \O \E \end{bmatrix} \qquad
    S_2 \defeq \begin{bmatrix} \C \E \M \O \E \O \\
      \E \C \O \M \O \E \end{bmatrix} \\
    S_3 \defeq \begin{bmatrix} \M \C \E \O \E \O \\
      \C \M \O \E \O \E \end{bmatrix} \qquad
    S_4 \defeq \begin{bmatrix} \E \E \M \O \C \O \\
      \E \E \O \M \O \C \end{bmatrix}
  \end{gather*}
\end{example}

\subsection{Description by Multivariate Polynomials}

When considering visual cryptography schemes, a common design goal is
to have reasonable small pixel expansion~$q$ while maximizing the
“contrast”, which is governed by the minimum over the sizes of all
$C_j$, those parts of basis matrices $S_j$ that are responsible for
identifying the colors.

A number of papers are devoted to construct schemes for black/white or
color visual cryptographic schemes optimizing these values, either by
structured constructions or by computer search, see, e.g., \cite{A14,
  DAR19, D+20, K22}.

In the present paper we focus on an elegant description of visual
cryptography schemes by means of multivariate polynomials
\cite[Sec.~3]{KI11}.  Each element of the color monoid
$M = \{ 0, c_1, \dots, c_k, 1 \}$ is considered an indeterminate,
where we denote the variable corresponding to $0 \in M$ (“black”)
as~$z$ and the variable corresponding to $1 \in M$ (“white”) as~$a$.

\begin{definition}
  A basis matrix~$S$ for~$n$ participants is encoded by a degree~$n$
  homogeneous polynomial~$p$ (with nonnegative integer coefficients)
  if each monomial $u x_1 \dots x_n$ is expanded to~$u$ copies of all
  possible $n!$ permutations of the column $[ x_1 \dots x_n ]^{\top}$.
  In this case~$p$ is referred to as the \emph{basis polynomial} for
  the basis matrix~$S$.
\end{definition}

\begin{figure}
  \begin{center}
    \medskip
    \raisebox{-.5\height}{\includegraphics[width=2cm]{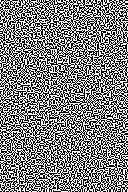}} $+$
    \raisebox{-.5\height}{\includegraphics[width=2cm]{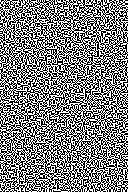}} $=$
    \raisebox{-.5\height}{\includegraphics[width=2cm]{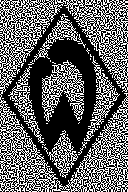}}
  \end{center}
  \caption{Example of a $(2, 2)$-threshold visual cryptography scheme.}
  \label{fig:svw}
\end{figure}

\begin{example}\label{exa:bw}
  Consider the degree~$2$ homogeneous polynomials
  \[ f_0 \defeq 2 a z \qquad\text{ and }\qquad f_1 \defeq a^2 + z^2 \,, \]
  with variables~$a$ for~$1$ and~$z$ for~$0$.  The corresponding
  basis matrices are
  \[ X_0 = \begin{bmatrix} 1 ~ 1 ~ 0 ~ 0 \\ 0 ~ 0 ~ 1 ~ 1 \end{bmatrix}
    \qquad\text{and}\qquad
    X_1 = \begin{bmatrix} 1 ~ 1 ~ 0 ~ 0 \\ 1 ~ 1 ~ 0 ~ 0 \end{bmatrix}. \]
  Since every column appears twice we can half the pixel expansion
  by using the reduced basis matrices
  \[ \tilde X_0 = \begin{bmatrix} 1 ~ 0 \\ 0 ~ 1 \end{bmatrix}
    \qquad\text{and}\qquad
    \tilde X_1 = \begin{bmatrix} 1 ~ 0 \\ 1 ~ 0 \end{bmatrix}. \]
  These matrices form the classical $(2, 2)$-threshold black and white
  visual cryptography scheme, see Fig.~\ref{fig:svw}.

  Note that the common derivative (see discussion below) of the
  polynomials is $\delta f_0 = 2 (a + z) = 2 a + 2 z = \delta f_1$.
\end{example}

\begin{example}
  The degree~$3$ homogeneous polynomial \[p \defeq m a^2 + c a z
    + y a z \] corresponds to the matrix
  \[ S = \begin{bmatrix}
      \M \E \E \M \E \E \, \C \C \E \E \O \O \, \Y \Y \E \E \O \O \\
      \E \M \E \E \M \E \, \E \O \C \O \C \E \, \E \O \Y \O \Y \E \\
      \E \E \M \E \E \M \, \O \E \O \C \E \C \, \O \E \O \Y \E \Y
    \end{bmatrix} . \]
\end{example}

For a matrix~$S$ encoded by a basis polynomial~$p$, the restriction
of~$S$ to any subset $A \subseteq [n]$ actually depends -- up to
column permutation -- only on the cardinality of~$A$.  Furthermore,
the degree~$\ell$ homogeneous polynomial~$p^{(\ell)}$ corresponding to
any subset $A \subseteq [n]$ of size~$\ell$ can be obtained from
applying the differential operator
\[ \delta \,\defeq\, d_a + d_{c_1} + \dots + d_{c_k} + d_z \]
multiple times to the polynomial~$p$, i.e.,
\[ p^{(\ell)} = \delta^{n-\ell} p \,. \]

By this observation the condition i) of a visual cryptographic scheme
can be succinctly expressed by requiring that $p^{(t - 1)} =
\delta^{n - t + 1} p_i$ be equal for all degree~$n$ homogeneous
polynomials~$p_i$ corresponding to~$S_i$.

\begin{example}
  For the above polynomial $p = m a^2 + c a z + y a z$ we have
  $\delta p = 2 m a + a^2 + c a + c z + a z + y a + y z + a z$,
  corresponding to the matrix
  \[ S[\{1, 2\}] = \begin{bmatrix}
      \M \E \E \M \E \E \C \C \E \E \O \O \Y \Y \E \E \O \O \\
      \E \M \E \E \M \E \E \O \C \O \C \E \E \O \Y \O \Y \E \\
    \end{bmatrix} . \]
\end{example}

\begin{figure}
  \begin{center}
    \medskip
    \raisebox{-.5\height}{\includegraphics[width=24mm]{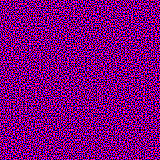}} $+$
    \raisebox{-.5\height}{\includegraphics[width=24mm]{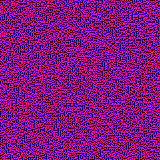}} $=$
    \raisebox{-.5\height}{\includegraphics[width=24mm]{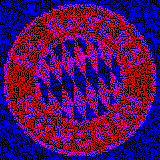}}
  \end{center}
  \caption{Example of a $(2, 2)$-threshold color visual cryptography scheme.}
  \label{fig:fcb}
\end{figure}

\begin{example}\label{exa:col}
  Take as color model the monoid $M = \{ 0, {\rm R}, {\rm B}, 1 \}$
  where ${\rm R} \wedge {\rm B} = 0$.   The basis matrices
  \[ S_{\rm R} \defeq \begin{bmatrix} \E \R \B \O \\
      \R \E \O \B \end{bmatrix} , ~
    S_{\rm B} \defeq \begin{bmatrix} \E \R \B \O \\
      \B \O \E \R \end{bmatrix} , ~
    S_{\rm 1} \defeq \begin{bmatrix} \E \R \B \O \\
      \E \R \B \O \end{bmatrix} \]
  satisfy property i) of a visual cryptography scheme for the colors
  ${\rm R}$, ${\rm B}$ and $1$, see Fig.~\ref{fig:fcb}.
  
  These can be encoded (after doubling the columns) by the basis
  polynomials $f_{\rm R} = 2 a r + 2 b z$, $f_{\rm B} = 2 a b + 2 r z$
  and $f_{\rm 1} = a^2 + r^2 + b^2 + z^2$.  Observe that their common
  derivative equals $2 (a + r + b + z)$.
\end{example}

Given a basis polynomial~$p$ for a visual cryptography scheme the
stacking process can be modeled by viewing~$p$ in the \emph{monoid
  ring}.  This means that the “variables” $c_0, \dots, c_{k+1}$ are
now considered as monoid elements, which are multiplied using the
monoid operation.  In this way, a monomial of the form $c_p c_q$ is
viewed as the monoid element $c_p * c_q$.

Looking at the black and white scheme of Example~\ref{exa:bw}, in the
monoid ring we have
\[ \tilde f_0 = 2 (a * z) = 2 z \quad\text{ and }\quad
  \tilde f_1 = a * a + z * z = a + z \,. \]
This describes the fact that stacking the shares for color~$0$
(“black”) gives all black, while for color~$1$ (“white”) we get half
white and half black.

Similarly, considering Example~\ref{exa:col} we obtain
\[ \tilde f_{\rm R} = 2 r + 2 z \,, \quad \tilde f_{\rm B} = 2 b + 2 z
  ~\text{ and }~ \tilde f_{\rm 1} = a + b + r + z \,, \]
describing the stacking of shares for colors $R$, $B$ and~$1$.

\section{Construction of Basis Matrices}%
\label{sec:construction}

We revisit the example constructions in \cite[Sec.~5]{KI11}.  There
the color monoid $\{ 0, \rm C, \rm Y, \rm M, 1 \}$ is employed, in
which the colors $\rm C$, $\rm Y$ and $\rm M$ are “independent” and
not related to each other, using only $1 * \rm X = \rm X$ and
$\rm X * \rm X = \rm X$ for some color~$\rm X$.  Our goal is to extend
these examples using the full color monoid $\{ 0, \rm Y, \rm M, \rm C,
\rm R, \rm G, \rm B, 1 \}$ (with $\rm Y * \rm M = \rm R$,
$\rm R * \rm G = 0$ and the like), or some submonoid thereof.

We aim to find alternative and generalizing constructions for some of
the following.  Here, the polynomial~$g_1$ corresponds to $f_{\rm M} -
f_{\rm C}$ and~$g_2$ corresponds to $f_{\rm M} - f_{\rm Y}$, where
$f_{\rm M}$, $f_{\rm C}$ and $f_{\rm Y}$ denote the respective basis
polynomials.  For a $(t, n)$-threshold scheme these choices satisfy
$\delta^{n - t - 1} g_1 = \delta^{n - t - 1} g_2 = 0$ (where~$\delta$
is the differential operator) in order to satisfy the security
property.

\begin{itemize}
\item In the $(n, n)$-threshold scheme the construction is based on
  the polynomials
  \[ g_1 = (m - c) (a - z)^{n-1} \text{ and }
    g_2 = (m - y) (a - z)^{n-1} \,. \]
\item Similarly, in the $(n \!-\! 1, n)$-threshold scheme it is used
  \[ g_1 = a (m - c) (a - z)^{n-2} \text{ and }
    g_2 = a (m - y) (a - z)^{n-2} \,. \]
\item For the $(2, n)$-threshold scheme the polynomials used are
  \begin{gather*}
    g_1 = (m - c) (a^{n-1} - z^{n-1}) \,,\\
    g_2 = (m - y) (a^{n-1} - z^{n-1}) \,.
  \end{gather*}
  A reduced pixel expansion is obtained by using “different
  permutation matrices”, by which repeated columns are contracted into
  one (similarly as in Example~\ref{exa:bw}).
\item Finally, the $(3, n)$-threshold scheme considers the polynomials
  \begin{gather*}
    g_1 = (m - c) \big( (n - 2) a^{n-1} - (n - 1) a^{n-2} z + z^{n-1} \big) , \\
    g_2 = (m - y) \big( (n - 2) a^{n-1} - (n - 1) a^{n-2} z + z^{n-1} \big) .
    \end{gather*}
\end{itemize}

The following result presents a construction for the color monoid
$\{ 0, \rm C, \rm M, \rm B, 1 \}$ where $\rm C * \rm M = \rm B$.

\begin{proposition}
  For any $n \ge 2$ there exists an $(n, n)$-threshold scheme for the
  color monoid $\{ 0, \rm C, \rm M, \rm B, 1 \}$, which is based on
  the polynomials
  \begin{gather*}
    g_1 = (m - a) (c - z) (a - z)^{n-2} \,, \\
    g_2 = (m - z) (c - a) (a - z)^{n-2} \,.
  \end{gather*}
\end{proposition}

\begin{proof}
  Using common rules for differentiation we obtain
  \[ \delta g_1 = (a - z)^{n-2} \delta (m - a) (c - z)
    + (m - a) (c - z) \delta (a - z)^{n-2} \]
  with $\delta (m - a) (c - z) = \delta (m c + a z - m z - c a)
  = m + c + a + z - m - z - c - a = 0$ and $\delta (a - z)^{n-2}
  = (n - 2) (a - z)^{n-3} \delta (a - z)$, where $\delta (a - z)
  = 1 - 1 = 0$.  Similarly, we compute $\delta g_2 = 0$, which
  establishes the security requirement.

  Now we may view~$g_1$ and~$g_2$ in the monoid ring in order to
  describe the stacking of shares for image recovery.  Since~$z$ is
  the absorbing element in the monoid, we obtain the non-black part by
  setting the variable~$z$ to~$0$.  This gives
  \[ g_1(0) = (m \!-\! a) c a^{n-2} \quad\text{ and }\quad
    g_2(0) = m (c \!-\! a) a^{n-2} \,, \]
  corresponding to basis submatrices encoded by
  \[ f_{\rm B}(0) = m c a^{n-2} \,, ~ f_{\rm C}(0) = c a^{n-1}
    \text{ and } f_{\rm 1}(0) = m a^{n-1} \,. \]
  Since $m * c = b$ we thus obtain the recovery property.
\end{proof}

\begin{example}
  In the case $n = 2$ we obtain $g_1 = (m - a) (c - z) = m c + a z
  - c a - m z$ and $g_2 = (m - z) (c - a) = m c + a z - m a - c z$
  (with $\delta g_1 = \delta g_2 = 0)$, so that
  \[ f_{\rm B} = m c + a z \,, \quad f_{\rm C} = c a + m z \,,
    \quad f_{\rm M} = m a + c z \,. \]
  The corresponding basis matrices are
  \[ S_{\rm B} = \begin{bmatrix} \M \C \E \O \\
      \C \M \O \E \end{bmatrix} , ~
    S_{\rm C} = \begin{bmatrix} \C \E \M \O \\
      \E \C \O \M \end{bmatrix} , ~
    S_{\rm M} = \begin{bmatrix} \M \E \C \O \\
      \E \M \O \C \end{bmatrix} , \]
  and the pixel expansion is only~$4$, compared to~$6$ in the
  construction of \cite[Thm.~3]{KI11}.
  
  For $n = 3$ we have $g_1 = (m - a) (c - z) (a - z) = m c a + a^2 z
  + c a z + m z^2 - c a^2 - m a z - m c z - a z^2$ and $g_2 = (m - z)
  (c - a) (a - z) = m c a + a^2 z + m a z + c z^2 - m a^2 - c a z - m
  c z - a z^2$, leading to the polynomials
  \begin{gather*}
    f_{\rm B} = m c a + a^2 z + c a z + m z^2 + m a z + c z^2 \,, \\
    f_{\rm C} = c a^2 + m a z + m c z + a z^2 + m a z + c z^2 \,, \\
    f_{\rm M} = m a^2 + c a z + m c z + a z^2 + c a z + m z^2 \,.
  \end{gather*}
\end{example}

We conclude with an example that also uses the color monoid
$\{ 0, \rm C, \rm M, \rm B, 1 \}$, but four basis matrices for each of
the colors ${\rm C}$, ${\rm M}$, ${\rm B}$ and~$1$ (“white”) are
provided.

\begin{example}
  If the colors $\rm C$, $\rm M$ and $\rm B$ are considered
  independent, an analogous result of \cite[Thm.~3]{KI11} produces the
  following $(2, 2)$-threshold scheme.  Here, employing the polynomials
  $g_1 = (m - c) (a - z) = m a + c z - m z - c a$, $g_2 = (m - b)
  (a - z) = m a + b z - m z - b a$ and $g_3 = (m - a) (a - z) =
  m a + a z - m z - a^2$ results in
  \begin{gather*}
    f_{\rm M} = m a + c z + b z + a z \,, \\
    f_{\rm C} = c a + m z + b z + a z \,, \\
    f_{\rm B} = b a + m z + c z + a z \,, \\
    f_{\rm 1} = a^2 + m z + c z + b z \,.
  \end{gather*}
  The corresponding basis matrices have a pixel expansion of~$8$.

  On the other hand, making use of the property $\rm C * \rm M = \rm B$
  in the color monoid, we can devise an alternative $(2, 2)$-threshold
  scheme es follows.  Indeed, consider the polynomials $g_1 = (m - z)
  (c - a) = m c + a z - m z - c a$, $g_2 = (m - z) (c - a) = m c + a z
  - c z - m a$ and $g_3 = (m - z) (c - z) + (z - a) (a - z) = m c + a z
  + a z - m z - c z - a^2$, which lead to
  \begin{gather*}
    f_{\rm B} = m c + a z + a z \,, \\
    f_{\rm C} = c a + m z + a z \,, \\
    f_{\rm M} = m a + c z + a z \,, \\
    f_{\rm 1} = a^2 + m z + c z \,.
  \end{gather*}
  The corresponding basis matrices are
  \begin{gather*}
    S_{\rm B} = \begin{bmatrix} \M \C \E \O \E \O \\
      \C \M \O \E \O \E \end{bmatrix} , \\
    S_{\rm C} = \begin{bmatrix} \C \E \M \O \E \O \\
      \E \C \O \M \O \E \end{bmatrix} , \\
    S_{\rm M} = \begin{bmatrix} \M \E \C \O \E \O \\
      \E \M \O \C \O \E \end{bmatrix} , \\
    S_{\rm 1} = \begin{bmatrix} \E \E \M \O \C \O \\
      \E \E \O \M \O \C \end{bmatrix} ,
  \end{gather*}
  and they have a smaller pixel expansion of just~$6$.
\end{example} \newpage

\section{Conclusion and Outlook}\label{sec:conclusion}

We study visual cryptography schemes with general color models.
Multivariate polynomial rings and derivations are an elegant method
for obtaining basis matrices.  The monoid ring is a suitable device
for describing the visual recovery.

The schemes of Koga and Ishihara do mostly not make full use of the
monoid operation.  We have adopted their method to strive for a better
image reconstruction.  Modeling colors as (non-idempotent) monoid
enables to also capture gray scale behavior.  The design of a good
scheme encorporating this idea is left as further research.

\end{document}